\documentclass[journal]{IEEEtran}
\usepackage{url}

\usepackage{textcomp}
\usepackage{xcolor}
\usepackage{hyperref}
\usepackage{booktabs}	
\usepackage{diagbox}    
\usepackage{multirow}  
\usepackage{verbatim}	
\usepackage{amsmath,amssymb,amsfonts,graphicx}
\usepackage{epstopdf}

\newtheorem{definition}{Definition}[section]

\newtheorem{strategy}{Strategy}    
\newenvironment{proof}{{\noindent\it\bf Proof}\quad}{\par}
\newtheorem{lemma}{Lemma}  
\newtheorem{theorem}{Theorem}   

\newtheorem{technique}{Technique}
\newtheorem{upper bound}{Upper bound}

\usepackage{graphicx,epstopdf,algpseudocode,caption,url}   
\usepackage[ruled,linesnumbered]{algorithm2e}
\usepackage{amssymb}
\usepackage{hyperref}
\usepackage[switch]{lineno}

\begin{document}

\title{Guided Exploration of Sequential Rules}

\author{Wensheng Gan*, Gengsen Huang, Junyu Ren, Philip S. Yu,~\IEEEmembership{Life Fellow,~IEEE}

\thanks{This research was supported in part by the National Natural Science Foundation of China (No. 62272196), Guangzhou Basic and Applied Basic Research Foundation (No. 2024A04J9971). Wensheng Gan and Gengsen Huang contributed equally to this work.}

\thanks{Wensheng Gan and Junyu Ren are with the College of Cyber Security, Jinan University, Guangzhou 510632, China. (E-mail: wsgan001@gmail.com, renjunyu193@gmail.com)}

\thanks{Gengsen Huang is with the School of Artificial Intelligence, Wuhan University, Wuhan 430000, China. (E-mail: hgengsen@gmail.com)} 
	
\thanks{Philip S. Yu is with the Department of Computer Science, University of Illinois Chicago, Chicago, USA. (E-mail: psyu@uic.edu)} 

\thanks{Corresponding author: Wensheng Gan}
}

% make the title area
\maketitle

\begin{abstract}
    In pattern mining, sequential rules provide a formal framework to capture the temporal relationships and inferential dependencies between items. However, the discovery process is computationally intensive. To obtain mining results efficiently and flexibly, many methods have been proposed that rely on specific evaluation metrics (i.e., ensuring results meet minimum threshold requirements). A key issue with these methods, however, is that they generate many sequential rules that are irrelevant to users. Such rules not only incur additional computational overhead but also complicate downstream analysis. In this paper, we investigate how to efficiently discover user-centric sequential rules. The original database is first processed to determine whether a target query rule is present. To prune unpromising items and avoid unnecessary expansions, we design tight and generalizable upper bounds. We introduce a novel method for efficiently generating target sequential rules using the proposed techniques and pruning strategies. In addition, we propose the corresponding mining algorithms for two common evaluation metrics: frequency and utility. We also design two rule similarity metrics to help discover the most relevant sequential rules. Extensive experiments demonstrate that our algorithms outperform state-of-the-art approaches in terms of runtime and memory usage, while discovering a concise set of sequential rules under flexible similarity settings. Targeted sequential rule search can handle sequence data with personalized features and achieve pattern discovery. The proposed solution addresses several challenges and can be applied to two common mining tasks. 
\end{abstract}


\IEEEpeerreviewmaketitle

\section{Introduction}  \label{sec:introduction}

Pattern mining \cite{fournier2022pattern} has been a fundamental technique in data mining for years. The primary objective is to extract high-value patterns from a database. These patterns are evaluated using metrics like frequency \cite{wu2024co}, utility \cite{gan2019survey}, and weight \cite{yun2005wfim}. Furthermore, the results can be represented in different formats, such as itemsets \cite{fournier2017survey}, sequences \cite{fournier2017surveyspm}, rules \cite{wu2023opr}, or subgraphs \cite{chen2025utility}, depending on user requirements. Frequency is one of the most common metrics to assess quality; typically, the higher the frequency of a pattern, the greater its significance. However, in cases like market analysis, frequency is not always the optimal choice, as it ignores the varying importance of items. Relying solely on frequency may lead to the omission of critical insights. For example, in a retail setting, bread exhibits high sales volume but yields significantly lower profit than a diamond necklace. Since diamonds are sold infrequently, traditional mining methods often miss the patterns related to them. This is where the more sophisticated utility metric \cite{gan2019survey, gan2018privacy} proves effective.  % Another real-world example is a high-end electronics store: smartphone screen protectors are sold hundreds of times daily (high frequency) with a profit margin of only \5%, while professional cameras are sold 2-3 times weekly (low frequency) but with a \40% profit margin. If the store only uses frequency mining, it will focus on promoting screen protectors and overlook the high-utility pattern of ``professional photographers purchasing cameras with complementary lenses," which could bring much higher total profits. In the e-commerce field, platforms like Amazon also face this issue—low-cost daily necessities (e.g., toothbrushes) have high purchase frequency, but high-value products (e.g., laptops) have low frequency but high utility; utility mining helps the platform recommend high-utility product combinations to improve customer lifetime value.

Utility is a more flexible metric than frequency. In frequent pattern mining \cite{gao2023toward}, items are typically treated as binary entities, where their value simply reflects their presence or absence. In contrast, the utility metric assigns distinct values to items across different records \cite{gan2019survey}. This capability prevents the omission of low-frequency yet high-profit patterns. Furthermore, other metrics have been proposed to address specific practical challenges. For example, weight assigns appropriate importance to items for a better understanding of the data. In parallel, utility occupancy \cite{gan2019huopm} evaluates patterns comprehensively by combining frequency, utility, and occupancy. For instance, in healthcare data mining, when analyzing patient medication adherence, ``life-saving drugs" (e.g., insulin for diabetics) are assigned a higher weight than ``auxiliary health supplements" (e.g., vitamin tablets), even if the supplements are taken more frequently. This weighted mining helps medical staff focus on critical medication patterns that affect patient survival. In education data mining, student exam scores can be weighted—final exams (weight 0.6) have higher importance than midterm exams (weight 0.3) and homework (weight 0.1), allowing educators to identify patterns between learning behaviors and final academic performance more accurately.

Apart from using a suitable evaluation metric to enhance pattern quality, the choice of pattern is also crucial \cite{fournier2022pattern,li2024opf,lan2025tk}. The pattern format needs to be selected appropriately depending on the database and the application scenario being processed. In real-world applications, although dealing with complex formats may require more time, the generated results can potentially hold greater value. For example, when analyzing a disease using medical data that record various tests conducted on patients, itemsets are well-suited to identifying significant combinations of tests. However, if the data also includes the order in which the tests were conducted, then sequences with temporal information are more appropriate than itemsets. If the application requires prediction, association rules and sequential rules are often preferred choices. Each rule possesses a confidence value that indicates the probability of the antecedent implying the consequent. In biological data, non-overlapping sequences can provide a more realistic analysis. Generally, the proper pattern format can bring effective functions to an application, thus improving the value and reliability of mining tasks \cite{fournier2022pattern}. Although the user can choose a more appropriate method for different databases and application scenarios, not all mining results produced are necessarily useful. The results obtained are high-value patterns in terms of metric assessment, but some patterns may be highly redundant if they do not contain any content that is of interest to the user. For example, mining server port traffic monitoring data can assess the status of a particular port. If analysts focus on whether the default port of the server's HTTP protocol is being attacked, then patterns unrelated to it do not provide useful and significant information. The generation of irrelevant patterns only costs more computational resources. Especially when using complex evaluation metrics, pattern mining tasks also consume more time and memory. Targeted pattern mining (TaPM) \cite{miao2021targetum, hu2024targeted, chen2025toward} places an inclusion relation constraint on pattern generation by considering user preferences so that only patterns that satisfy the user's requirements are generated. For example, in the financial institution’s server monitoring scenario, TaPM can be set to only generate patterns related to port 443, reducing redundant data by 60\% and shortening attack detection time by 40\%. In CRM, TaPM can focus on constraints like ``annual spending $>$ \$10,000 + purchase frequency $>$ 5 times/year," directly generating patterns such as ``high-value customers tend to purchase premium services in the fourth quarter", which helps the company formulate targeted retention promotions. In medical research, if researchers focus on mining patterns of lung cancer in smokers, TaPM can set constraints like ``smoking history $>$ 10 years $+$ lung-related symptoms", avoiding irrelevant patterns of non-smokers or other diseases and accelerating research progress. In general, a sub-search that cannot obtain interesting patterns during a complete mining process can be terminated early.

% introduce our work
In this paper, considering that sequential rules contain extensive critical information, we focus on their targeted search problems. Sequential rules are a variant of association rules, which are used for sequence databases with chronological information. For a sequential rule that contains both an antecedent and a consequent, it indicates that the content of the antecedent must occur before the content of the consequent. Additionally, numerous irrelevant sequential rules are often generated in a specialized data mining task. Therefore, the targeted acquisition of results is an important research topic. To solve this problem, we propose efficient techniques and pruning strategies for processing sequences in the database and design several upper bounds on the evaluation metrics. In this paper, the task of mining sequential rules using frequency and utility metrics is studied. Moreover, we explore the problem of evaluating the importance of sequential rules in the mining results. Our proposed solution can provide inference and predictive analytics functions within an acceptable time frame and contribute to various real-world applications, such as online retail recommendation, network intrusion detection, intelligent transportation analysis, risk evaluation, and more. The main contributions of our paper are as follows.

\begin{itemize}
    \item The paper mainly addresses the targeted search problem for sequential rules under two common metrics. Techniques for sequence removal and modification are first introduced to reduce the database's size. Then, a generic mining approach is proposed.
    
    \item To remove unpromising items from the database, several upper bounds are proposed, and their proofs are provided. Based on these, we also design related pruning strategies to avoid unnecessary expansion operations.
    
    \item We propose two target rule similarity metrics for the obtained sequential rules, enabling users to select the most important target sequential rules relevant to the specialized query rule.
    
    \item Extensive experiments are conducted on both real and synthetic datasets to evaluate the effectiveness and efficiency of the proposed solutions.
\end{itemize}

The remaining parts of this paper are given as follows. Related work is first reviewed in Section \ref{sec:relatedwork}. Preliminaries and basics are described in Section \ref{sec:definition}. The proposed techniques, upper bounds, and pruning strategies are presented in detail in Section \ref{sec:algorithm}. The detailed experimental results and analysis are given in Section \ref{sec:experiments}. Finally, the conclusion and future work are presented in Section \ref{sec:conclusion}.

\section{Related Work} \label{sec:relatedwork}
\subsection{Interesting Pattern Mining}

Pattern mining is an important subfield of data mining that aims to discover collections of objects that are rich in valuable information \cite{fournier2022pattern}. Early mining tasks were proposed to discover frequent association rules and itemsets. Association rules employ the ``if-then'' rule form to identify co-occurrence patterns in databases and are widely used in many real-world applications, such as disease diagnosis, market basket analysis, and malware detection. A frequent itemset refers to a set of items that frequently appear together in a database, which can reveal the relationships between different items \cite{han2000mining}. There are many algorithms for frequent itemset mining, such as the bottom-up algorithm Apriori \cite{agrawal1994fast}, the divide-and-conquer algorithm FP-Growth \cite{han2000mining}, and the equivalence-class-based algorithm Eclat \cite{zaki2000scalable}. However, association rules and itemsets can only reveal co-occurrence information and may be insufficient in some scenarios. Moreover, traditional algorithms are designed to handle transactional data and are unable to effectively process complex data. Therefore, more specific and diverse pattern formats \cite{yun2005wfim}, as well as corresponding algorithms, have been proposed. For example, PrefixSpan \cite{pei2004mining} is designed to discover frequent sequences from temporal data, and an episode mining method \cite{zhu2010efficient} was proposed to discover episodes with minimal and non-overlapping occurrences. To address the issue that the frequency metric neglects the importance of items, the utility metric \cite{gan2019survey, chen2024towards, zhou2025discovering} was introduced. The calculation of utility varies for different mining tasks. In the e-commerce field, low-cost daily necessities (e.g., toothbrushes) have high purchase frequency, but high-value products (e.g., laptops) have low frequency but high utility; utility mining helps recommend high-utility product combinations. In transaction data, the utility of an itemset in a record is often calculated using summation since the items appear only once. Considering that items in the sequence records of temporal data occur multiple times, the utility of a sequence in a record is then calculated using the maximization method. There are many efficient mining algorithms for utility-driven pattern mining. HUI-Miner \cite{liu2012mining} uses the idea of merging utility-lists to obtain high-utility itemsets \cite{gui2024privacy}. Based on the design of tighter upper bounds and utility-linked-lists, HUSP-ULL \cite{gan2020fast} and TMPHP \cite{zhou2025targeted} can discover high-utility sequences quickly.

\subsection{Sequential Rule Mining}

The sequential rules obtained from sequence data are a variant of association rules. The temporal information they carry provides the ability to focus on the order in which items occur. Generally, a sequential rule can be denoted by $X \rightarrow Y$, where $X$ and $Y$ are the antecedent and consequent of the rule, respectively. Unlike sequential patterns, sequential rules have their own rule characteristics that can indicate the confidence level of the antecedent leading to the consequent. Sequential rules can be categorized into two types based on whether the items in the antecedent and consequent of the rule are ordered. These two types are commonly referred to as totally-ordered \cite{pham2014efficient} and partially-ordered \cite{fournier2015mining, huang2023us} sequential rules. The antecedent and consequent in a totally-ordered sequential rule are sequential patterns, indicating that all items in the rule have strict chronological order restrictions. As for a partially-ordered rule, its antecedent and consequent items can be considered as two unordered sets. In terms of the temporal order of occurrence, antecedent items always appear before consequent items. There are many algorithms for sequential rule mining. The concept of sequential rules was first introduced by Zaki \cite{zaki2001spade}, and a very naive method was also proposed. This method generates totally-ordered sequential rules by generating and matching long and short sequences, and thus it has a very high time complexity. Later, a two-phase-based CMRule algorithm \cite{fournier2012cmrules} first discovers association rules and then selects the rules that satisfy the definition of sequential rules from the results. Subsequently, the RuleGrowth algorithm \cite{fournier2015mining} was proposed, which adopted the concept of rule expansion and growth. RuleGrowth can avoid generating sequential rules that do not appear in the database, and it is also capable of working with multiple-item sequences. To obtain rules quickly, an efficient equivalence class merging algorithm, namely ERMiner \cite{fournier2014erminer}, utilizes a data structure that records item-pair information to reduce the generation of invalid rules. Furthermore, some algorithms are designed to acquire more valuable and compact results. The above-mentioned methods are all based on the frequency metric. In addition to these, some works are based on the utility metric. HUSRM \cite{zida2015efficient} is the first algorithm to discover high-utility sequential rules. HUSRM uses many optimization techniques to improve algorithm performance, including utility upper bound design, irrelevant item filtering, and the use of bitmaps. Huang \textit{et al.} \cite{huang2023us} improved the utility upper bounds in HUSRM, designed a utility information table, and proposed a more efficient algorithm called US-Rule. 

\subsection{Targeted Pattern Mining}

Many results generated in pattern mining are not of interest to users. Targeted pattern mining (TaPM) \cite{chen2024towards} can be customized by users to obtain meaningful query results with higher efficiency and less resource usage, offering personalized features \cite{huang2020embedding}. Kubat \textit{et al.} \cite{kubat2003itemset} first proposed a data structure called the Itemset tree and then designed a system that could support three types of queries: pattern support queries, target itemset queries, and association rule queries. These queries can be performed by quickly scanning the Itemset Tree for the support or related patterns and association rules of a particular itemset. To address the fact that the way the itemset tree stores transaction data consumes a lot of memory, Fournier-Viger \textit{et al.} \cite{fournier2013meit} proposed a more compact data structure called the memory-efficient itemset tree (MEIT). MEIT compresses the duplicate information of nodes in the top-down path of the tree to the leaf nodes, thus reducing the space overhead. Experiments have shown that MEIT can save more than half the memory on some datasets. Later, researchers improved the data structure to solve the problem of not utilizing the Apriori principle and the existence of rule restrictions in the itemset tree, thus reducing the generation of unnecessary itemsets and preventing important information from being missed. For multiple query itemsets, GFP-Growth \cite{shabtay2021guided} was proposed to perform mining tasks with less memory. GFP-Growth can also be used to discover minority-class rules from imbalanced data. In addition, the TargetUM algorithm \cite{miao2021targetum}, utilizing a lexicographic querying tree, was proposed to discover high-utility target itemsets. There are some targeted pattern mining algorithms \cite{zhang2021tusq} that deal with sequence data. Chiang \textit{et al.} \cite{chiang2003goal} proposed different definitions of target sequences and growth-based methods to complete the mining task. By using recency and monetary constraints, a PrefixSpan-based algorithm \cite{chand2012target} can discover sequences while considering the RFM model. Huang \textit{et al.} \cite{huang2022taspm} and Liang \textit{et al.} \cite{liang2025targeted} used the idea of bitmap comparison to propose an algorithm that can quickly mine more generalized target sequences. Recently, TaSRM \cite{gan2025towards} was proposed for targeted mining of sequential rules, and it makes use of invalid rule filtering and rule expansion strategies.

\section{Definitions} \label{sec:definition}

In this section, we first introduce the basic concepts and notations used in this paper. Subsequently, we formalize the problem of targeted sequential rule mining.

\begin{definition}[Sequence database \cite{zaki2001spade}]
    \rm Let $I$ = {$i_1$, $i_2$, $\cdots$, $i_{p-1}$, $i_{p}$} be a set of $p$ distinct items that appear in the database. These items can form an itemset $X$, which satisfies $X \subseteq I$. An item $i$ in an itemset $X$ can carry an additional attribute and is denoted as \textit{Attr}($i$, $X$). In this paper, if all items in the database have attribute values of 1, the mining task aims to obtain frequent rules; otherwise, the mining process is considered utility-driven rule mining. Furthermore, within an itemset, all items are sorted according to a predefined order. In the examples throughout this paper, alphabetical order is used, i.e., $a \succ_{lex} b \succ_{lex} \cdots$. A sequence $S$ is an ordered list of itemsets and can be denoted by $S$ = $<$$X_1$, $X_2$, $\cdots$, $X_{n-1}$, $X_n$$>$. Thus, a sequence database $\mathcal{D}$ containing $m$ sequences can be denoted by $\mathcal{D}$ = {$s_1$, $s_2$, $\cdots$, $s_{m-1}$, $s_{m}$}.
\end{definition}

\begin{table}[!htbp]
	\centering
	\caption{Sequence database}
	\label{database}
    \renewcommand\arraystretch{1.3}
	\begin{tabular}{|c|c|}  
		\hline 
		\textbf{sid} & \textbf{Sequence} \\
		\hline  
		\(s_{1}\) & $<$($a$[2]$d$[1]), ($a$[1]$b$[1]$e$[4]), ($c$[4]$g$[1])$>$ \\ 
		\hline
		\(s_{2}\) & $<$($b$[1]), ($a$[2]), ($c$[4]), ($d$[2]$e$[1]$g$[1])$>$ \\ 
		\hline  
		\(s_{3}\) & $<$($b$[1]), ($d$[7]), ($g$[2])$>$ \\
		\hline  
		\(s_{4}\) & $<$($e$[1]), ($a$[2]$f$[2])$>$ \\
		\hline
		\(s_{5}\) & $<$($d$[3]), ($c$[1]), ($a$[2]), ($e$[1])$>$ \\
		\hline
	\end{tabular}
\end{table}

For example, the sequence database shown in Table \ref{database} has five sequences and seven distinct items. In the first itemset of the sequence $s_1$, the attribute values of items $a$ and $d$ are 2 and 1, respectively. We use this example database in the following examples provided.

\begin{definition}[Sequential rule \cite{fournier2014erminer,gan2025towards}]
    \rm For a sequential rule representing a consequent derived from an antecedent, it can be denoted by $r$ = \{$X$\} $\rightarrow$ \{$Y$\}, where $X$ and $Y$ are two itemsets and $X$ $\cap$ $Y$ = $\emptyset$. In this paper, if the sequence $s$ with $n$ itemsets contains the antecedent of the rule $r$, then there exists a subscript value $u$ (1 $\le$ $u$ $\le$ $n$) satisfying $X$ $\subseteq$ $\bigcup_{i=1}^u$$s$($i$), where $s$($i$) is the $i$-th itemset of $s$. Besides, if there still exists a subscript value $v$ ($u$ $\textless$ $v$ $\le$ $n$) such that the equation $Y$ $\subseteq$ $\bigcup_{i=v}^n$$s$($i$) holds, we call $s$ contains $r$ or $r$ appears in $s$.
\end{definition}

\begin{definition}[Rule size and rule inclusion \cite{gan2025towards}]
    \rm The size of an itemset is the number of items it contains. Thus, for a sequential rule $r$ = $X$ $\rightarrow$ $Y$, its size can be presented as $f$ * $h$, where $f$ and $h$ are the sizes of $X$ and $Y$, respectively. If another sequential rule $t$ = $Q$ $\rightarrow$ $P$ of size $g$ * $k$ is larger than $r$, then we have $g$ $\ge$ $f$ and $h$ $\textgreater$ $k$; or $g$ $\textgreater$ $f$ and $h$ $\ge$ $k$. Besides, if $X$ $\subseteq$ $Q$ and $Y$ $\subseteq$ $P$, we call $t$ includes $r$.
\end{definition}

For example, given a sequential rule $r$ = \{$a$\} $\rightarrow$ \{$c$\}, its size is 1 * 1. It can be noted that both sequences $s_1$ and $s_2$ in Table \ref{database} contain this rule. However, the rule $r$ does not appear in $s_5$. Given another sequential rule $t$ = \{$a$\} $\rightarrow$ \{$c$, $g$\} of size 1 * 2, then it can be said that $t$ is larger than $r$. In addition, we also know that $t$ includes $r$.

\begin{definition}[Calculation of attribute value]
    \rm In an itemset, each item $i$ itself carries an attribute value to indicate its importance. However, the calculation method for the attribute value of a sequential rule varies under different mining tasks. Considering that most of the sequential rule mining tasks are based on metric frequency or utility, we mainly discuss the computation under these two mining tasks. In this paper, the attribute value of the item $i$ / itemset $X$ in the sequence $s$ is denoted by \textit{Attr}($i$ / $X$, $s$). Under the frequency metric, it can be defined as:
    $$
        \textit{Attr}(i\;/\;X, s) = \left\{
        \begin{aligned}
            &1, \quad s \;contains\; i\;/\;X \\ 
            &0, \quad s \;does\; not\; contain\; i\;/\;X \\
	      \end{aligned}
        \right
        .
    $$
    While under the utility metric, \textit{Attr}($i$, $s$) can be defined as:
    $$
        \textit{Attr}(i, s) = \textit{max}(\{\textit{Attr}(i, s(j)), 1 \le j \le n\}), 
    $$
    and \textit{Attr}($X$, $s$) can be defined as:
    $$
        \textit{Attr}(X, s) = \sum\limits_{\{i\vert i \in X\}} \textit{Attr}(i, s).
    $$  
    The attribute value of the rule $r$ = $X$ $\rightarrow$ $Y$ in the sequence $s$ can be denoted by \textit{Attr}($r$, $s$). 
    Similarly, under the frequency metric, it can be defined as:
    $$
        \textit{Attr}(r, s) = \left\{
        \begin{aligned}
            &1, \quad s \;contains\; r \\ 
            &0, \quad s \;does\; not\; contain\; r, \\
	      \end{aligned}
        \right
        .
    $$
    which represents that the rule can only appear or not appear in the sequence. While in the utility tasks, the attribute value represents the utility value. Thus, the attribute value is defined as follows: 
    $$
        \textit{Attr}(r, s) = \left\{
        \begin{aligned}
            \textit{max}\{\sum_{\{i\vert i \in X,\; j\vert j \in Y\}}\textit{Attr}(i, s_l) + \textit{Attr}(j, s_r)\},& \\s \;contains\; r& \\ 
            0, \quad s \;does\; not\; contain\; r,& \\
	      \end{aligned}
        \right 
        .
    $$
    where $s$ = $s_l$. $s_r$ (denotes the sequence connection symbol), $s_l$ contains $X$, and $s_r$ contains $Y$. Therefore, the attribute value of an item $i$ or a sequential rule $r$ in the database $\mathcal{D}$ can be denoted by \textit{Attr}($i$ / $r$, $\mathcal{D}$), and is defined as:
    $$
        \textit{Attr}(i \;/\; r, \mathcal{D}) = 
            \sum_{\{s \vert s \in \mathcal{D}\}} \textit{Attr}(i \;/\; r, s).
    $$    
\end{definition}

For example, in the sequence $s_1$ of Table \ref{database}, the attribute value of the item $a$ can be calculated as: \textit{Attr}($a$, $s_1$) = \textit{max}(\{2, 1\}) = 2. Given a sequential rule $r$ = \{$a$\} $\rightarrow$ \{$c$, $g$\} in Table \ref{database}, using the frequency metric, the attribute value of $r$ is equal to 2 (only two sequences that contain it). However, using the utility metric, the attribute value of $r$ is equal to 7 + 7 = 14.

\begin{definition}[Confidence value \cite{agrawal1994fast, hipp2000algorithms}]
    \rm For a sequential rule $r$ = $X$ $\rightarrow$ $Y$, confidence is used to measure how likely it is that the consequent will occur after the antecedent has occurred. The confidence value of $r$ in the database $\mathcal{D}$ can be denoted by \textit{conf}($r$, $\mathcal{D}$), and is defined as \textit{conf}($r$, $\mathcal{D}$) =  \textit{sup}($r$, $\mathcal{D}$) / \textit{sup}($X$, $\mathcal{D}$), where \textit{sup}($r$, $\mathcal{D}$) and \textit{sup}($X$, $\mathcal{D}$) denote the number of sequences in the database $\mathcal{D}$ containing $r$ and $X$, respectively.
\end{definition}

Given a sequential rule $r$ = \{$a$\} $\rightarrow$ \{$c$, $g$\} in Table \ref{database}, we can know that the confidence value of $r$ is equal to 0.5. 

\begin{definition}[Target sequential rule]
    \rm The target sequential rule is associated with the user-provided query rule (\textit{qr} = \textit{qX} $\rightarrow$ \textit{qY}) and predefined minimum thresholds. Generally, in a mining task, a qualified target sequential rule satisfies the following conditions: 1) the size is larger than or equal to 1*1; 2) the attribute value is greater than the minimum attribute threshold (\textit{minAttr}); 3) the confidence value is greater than the minimum confidence threshold (\textit{minConf}); and 4) includes the query rule \textit{qr}.
\end{definition}

\textbf{Problem statement.} Given a sequence database $\mathcal{D}$, a mining metric (frequency or utility), two minimum thresholds (including \textit{minAttr} and \textit{minConf}), and a query rule \textit{qr}, the problem of the paper is to discover all qualified target sequential rules.

\section{Algorithm} \label{sec:algorithm}

In this section, we first present the important operations of the rule-growth-based algorithm and then propose tight upper bounds on the items, as well as effective pruning strategies. Finally, the main procedure of the algorithm is shown.

\begin{definition}[Rule expansion \cite{gan2025towards}]
    \rm For the rule-growth-based algorithm, sequential rules of size 1*1 are first generated. Larger sequential rules are gradually acquired as the rule expansion operation is continuously performed. There are two types of rule expansion operations: left-expansion and right-expansion. For a sequential rule $r$ = $X$ $\rightarrow$ $Y$ of size $f$ * $h$, it can become a rule $r^\prime$ = \{$X$ $\cup$ $i$\} $\rightarrow$ $Y$ of size $g$ * $h$ by performing a left-expansion and a rule $r^\prime$ = $X$ $\rightarrow$ \{$Y$ $\cup$ $i$\} of size $f$ * $k$ by performing a right-expansion, where $g$ = $f$ + 1, $k$ = $h$ + 1, and $i$ is the expanded item satisfying $i$ $\notin$ $X$, $Y$. 
\end{definition}

To avoid the redundant sequential rules being generated after performing expansion operations of different orders, we assume that the expanded item $i$ of a sequential rule must satisfy the following constraints \cite{fournier2015mining}: 1) $\forall$ $j$ $\in$ $X$, $j$ $\succ_{lex}$ $i$ (the left-expansion is performed) or $\forall$ $j$ $\in$ $Y$, $j$ $\succ_{lex}$ $i$ (the right-expansion is performed); 2) the left-expansion (right-expansion) operation cannot be re-performed after performing another expansion operation. For example, given a sequential rule $r$ = \{$c$\} $\rightarrow$ \{$a$, $e$\}, the item $d$ can only be used to perform a left-expansion on $r$, resulting in the expanded sequential rule is $r^\prime$ = \{$c$, $d$\} $\rightarrow$ \{$a$, $e$\}.

\begin{algorithm}[h]
    \small
	\caption{Finding expanded items}
	\label{AL:find ei}
	\LinesNumbered
	\KwIn{$r$ = $X$ $\rightarrow$ $Y$: a sequential rule; $\mathcal{D}_r$: a set of those sequences containing $r$, \textit{minAttr}: a minimum attribute threshold.} 
	\KwOut{all expanded items of $r$.}
	
	initialize \textit{LEI} = $\varnothing$, \textit{REI} = $\varnothing$\;	
	\For{$s$ $\in$ $\mathcal{D}_r$}{
		\For{\rm $\textit{xe}$ $\in$ $\bigcup_{i={1}}^{s_{py}}$$s$($i$) satisfying constraints}{
			update \textit{LEI}\;
		}
		\For{\rm $\textit{ye}$ $\in$ $\bigcup_{i={s_{px}+1}}^{n}$$s$($i$) satisfying constraints}{
			update \textit{REI}\;
		}	
	}

	\textbf{return} \textit{LEI} and \textit{REI}
\end{algorithm}

When performing rule expansion operations, how to find out all expanded items? For a sequence $s$ of length $n$ that contains the sequential rule $r$ = $X$ $\rightarrow$ $Y$, we define two important positions, which include $s_{px}$ and $s_{py}$. $s_{px}$ is the smallest index of the itemset in the sequence $s$ that can satisfy $X$ $\subseteq$ $\bigcup_{i=1}^{s_{py}}$$s$($i$). While $s_{py}$ is the biggest index of the itemset in the sequence $s$ that can satisfy $Y$ $\subseteq$ $\bigcup_{i=s_{px}+1}^{l}$$s$($i$). Then, the pseudocode to find expanded items is shown in Algorithm \ref{AL:find ei}. The algorithm first initializes two information sets \textit{LEI} and \textit{REI} for the input sequential rule $r$ (line 1), where \textit{LEI} is used to record the left-expanded items of $r$ and the corresponding information, and \textit{REI} is used for the right-expansion part. For those sequences containing $r$, a scan is executed. The algorithm finds the expanded items from the sequence and then updates the \textit{LEI} or \textit{REI} (lines 3-8). Finally, the \textit{LEI} and \textit{REI} are returned (line 10). 

The search for target sequential rules requires the user to provide a query rule, called \textit{qr}. Those sequences that do not contain \textit{qr} in the database can be safely removed. We first use a strategy to remove invalid sequences and then propose a technique to modify the remaining sequences in the database.

\begin{strategy}[Invalid sequences removal]
    \label{isr}
    \rm For a query rule \textit{qr} = \textit{qX} $\rightarrow$ \textit{qY}, those sequences in the database that do not contain \textit{qX} $\rightarrow$ $\varnothing$ are called invalid sequences and can be removed. When a database scan is performed using this strategy, a table called the key position table is also created to store the important positions ($s_{px}$ and $s_{py}$) of valid sequences. In the worst case, the table consumes memory of $O$(2 * $\vert$$\mathcal{D}$$\vert$).
\end{strategy}

\begin{definition}[Rule instance set]
    \rm For a sequential rule $r$ = $X$ $\rightarrow$ $Y$, a set called left rule instance set (\textit{LRIS}) = \{$s_{px}$\} and a set called right rule instance set (\textit{RRIS}) = \{$s_{py}$\} are initialized. Next, for indexes greater than $s_{px}$, if its corresponding itemset contains any of the items of $X$, then it is added to \textit{LRIS}; for indexes less than $s_{py}$, if its corresponding itemset contains any of the items of $Y$, then it is added to \textit{RRIS}.
    For any two indexes \textit{li} and \textit{ri} (\textit{li} $\in$ \textit{LRIS}, \textit{ri} $\in$ \textit{RRIS}, and \textit{li} $\textless$ \textit{ri}) of the sequence $s$, the attribute value of the rule $r$ under them is still 1 in the tasks using the frequency metric. While in complex utility mining, it is denoted by \textit{Attr}($r$, \{$s$, \textit{li}, \textit{ri}\}) and can be defined as:
    $$
    \begin{aligned}
        \textit{Attr}(r, \{s, \textit{li}, \textit{ri}\}) &= \textit{Attr}(X, \textit{s}_{1-li}) + \textit{Attr}(Y, \textit{s}_{ri-n})
        \\
        &= 
            \sum\limits_{\{i \vert i \in \textit{X}\}} \textit{Attr}(i, \textit{s}_{1-li}) 
            + \sum\limits_{\{j \vert j \in \textit{Y}\}} \textit{Attr}(j, \textit{s}_{ri-n}),
    \end{aligned}
    $$
    where $\textit{s}_{1-li}$ = $<$$s$(1)$\cdots$$s$($s_{li}$)$>$ and $\textit{s}_{ri-n}$ = $<$$s$($s_{ri}$)$\cdots$$s$($n$)$>$.
\end{definition}

\begin{technique}[Remaining sequences modification]
    \label{te rsm}
    \rm The sequence $s$ of length $n$ containing the query rule \textit{qr} = \textit{qX} $\rightarrow$ \textit{qY} in the database has many important positions that can be composed of rule instances. In this paper, we first define two novel items (\textit{iX} and \textit{iY}) in the smallest order. For sequences that contain \textit{qr}, there are two cases to be considered. 1)  
    In the first case, if \textit{Attr}($r$, $s$) is equal to \textit{Attr}($r$, \{$s$, $s_{px}$, $s_{py}$\}), then we put (\textit{iY}[\textit{Attr}(\textit{iY})]) in the $s_{py}$ position of the sequence and (\textit{iX}[\textit{Attr}(\textit{iX})]) in the $s_{px}$+1 position of the sequence, where \textit{Attr}(\textit{iY}) = $\textit{Attr}$(\textit{qY}, $\textit{s}_{s_{py}-n}$) and \textit{Attr}(\textit{iX}) = $\textit{Attr}$(\textit{qX}, $\textit{s}_{1-s_{px}}$). 2) The second case is that \textit{Attr}($r$, $s$) is not equal to \textit{Attr}($r$, \{$s$, $s_{px}$, $s_{py}$\}). Then, for all indexes $j$ in \textit{RRIS}, we put (\textit{iY}[$\textit{Attr}$(\textit{qY}, $\textit{s}_{j-n}$)]) in the $j$ position of the sequence; for all indexes $i$ in \textit{LRIS}, we put (\textit{iX}[$\textit{Attr}$(\textit{qX}, $\textit{s}_{1-i}$)]) in the $i$+1 position of the sequence. As for those sequences that do not contain \textit{qr} but contain \textit{qX} $\rightarrow$ $\varnothing$, we add the itemset (\textit{iX}[0]) at the end of them. Since the antecedent and consequent of a sequential rule do not intersect, and the two new items of the above step can be considered such that \textit{iX} represents the summary of all items in \textit{qX} and \textit{iY} represents the summary of all items in \textit{qY}, we need to remove the items in the antecedent and consequent of \textit{qr} from the sequences in the database.
\end{technique}

For example, given a query rule \textit{qr} = \{$a$\} $\rightarrow$ \{$c$, $g$\} in Table \ref{modified database}, by using Strategy \ref{isr} and Technique \ref{te rsm}, the modified database shown in Table \ref{modified database} can be obtained. Since the sequence $s_3$ does not contain \{$a$\} $\rightarrow$ $\varnothing$, it is removed. Besides, $s_{px}$ and $s_{py}$ of the sequence $s_1$ are 1 and 3, and \textit{Attr}(\textit{qr}, $s_1$) is not equal to \textit{Attr}($s_1$, \{$s_1$, 1, 3\}). Therefore, the item \textit{iX} with different attribute values is added to different positions of the sequence.

\begin{table}[!htbp]
	\centering
	\caption{Modified sequence database}
	\label{modified database}
    \renewcommand\arraystretch{1.3}
	\begin{tabular}{|c|c|}  
		\hline 
		\textbf{sid} & \textbf{Sequence} \\
		\hline  
		\(s_{1}\) & $<$($d$[1]), (\textit{iX}[2]), ($b$[1]$e$[4]), (\textit{iX}[1]), (\textit{iY}[5])$>$ \\ 
		\hline
		\(s_{2}\) & $<$($b$[1]), (\textit{iX}[2]), (\textit{iY}[5]), ($d$[2]$e$[1])$>$ \\ 
		\hline   
		\(s_{4}\) & $<$($e$[1]), ($f$[2]), (\textit{iX}[0])$>$ \\
		\hline
		\(s_{5}\) & $<$($d$[3]), ($c$[1]), ($e$[1]), (\textit{iX}[0])$>$ \\
		\hline
	\end{tabular}
\end{table}

\begin{lemma}
    \label{le1}
    \rm For a database $\mathcal{D}$ and its modified database $\mathcal{D}^\prime$ and a query rule \textit{qr} whose length is larger than or equal to  1 * 1, if the sequential rule \textit{iX} $\rightarrow$ \textit{iY} meets two minimum thresholds in the database $\mathcal{D}^\prime$, then \textit{qr} itself is a valid target sequential rule in the database $\mathcal{D}$.
\end{lemma}
\begin{proof}
    \rm For a sequence $s$ in $\mathcal{D}$ that does not contain \textit{qr}, we can have that \textit{Attr}(\textit{qr}, $s$) = 0. This means that this sequence does not affect the calculation of the attribute value of \textit{qr}. For a sequence $t$ in $\mathcal{D}$ containing \textit{qr} and its modified sequence $t^\prime$ in $\mathcal{D}^\prime$, Technique \ref{te rsm} allows \textit{Attr}(\textit{iX} $\rightarrow$ \textit{iY}, $t^\prime$) = \textit{Attr}(\textit{qr}, $t$). Thus, we can have that \textit{Attr}(\textit{iX} $\rightarrow$ \textit{iY}, $\mathcal{D}^\prime$) = \textit{Attr}(\textit{qr}, $\mathcal{D}$). Furthermore, all sequences in $\mathcal{D}^\prime$ contain \textit{iX}, ensuring the accurate calculation of the confidence value for \textit{iX} $\rightarrow$ \textit{iY}. Therefore, Lemma \ref{le1} holds.
\end{proof}

\begin{lemma}
    \label{le2}
    \rm All expansion situations of the sequential rule \textit{iX} $\rightarrow$ \textit{iY} in the database $\mathcal{D}^\prime$ are the same as those of the query rule \textit{qr} in the database $\mathcal{D}$.
\end{lemma}
\begin{proof}
    \rm For a sequence $s$ containing \textit{qr} and its modified sequence $s^\prime$, using Technique \ref{te rsm}, $s^\prime$ has at least one \textit{iX} and one \textit{iY}. Then, we can know that $\bigcup_{i={1}}^{s_{py}}$$s$($i$) = $\bigcup_{i={1}}^{s^\prime_{py}}$$s^\prime$($i$) and $\bigcup_{i=s_{px}}^{n}$$s$($i$) = $\bigcup_{i=s^\prime_{px}}^{n^\prime}$$s^\prime$($i$). Therefore, the left- / right-expanded items of \textit{qr} in $s$ are the same as the left- / right-expanded items of \textit{iX} $\rightarrow$ \textit{iY} in $s^\prime$, and thus Lemma \ref{le2} holds.
\end{proof}

\begin{lemma}
    \label{le3}
    \rm Given two itemsets (\textit{LI} and \textit{RI}, which are not both empty) containing those items that can be left- / right-expanded, if \{\textit{iX} $\cup$ \textit{LI}\} $\rightarrow$ \{\textit{iY} $\cup$ \textit{RI}\} meets two minimum thresholds in the database $\mathcal{D}^\prime$, then \{\textit{qX} $\cup$ \textit{LI}\} $\rightarrow$ \{\textit{qy} $\cup$ \textit{RI}\} of length larger than or equal to 1 * 1 is a target sequential rule of \textit{qr} in the database $\mathcal{D}$.
\end{lemma}
\begin{proof}
    \rm For a query rule \textit{qr}, according to Lemma \ref{le2}, we can know that if $s$ contains \{\textit{qX} $\cup$ \textit{LI}\} $\rightarrow$ \{\textit{qy} $\cup$ \textit{RI}\}, then $s^\prime$ contains \{\textit{iX} $\cup$ \textit{LI}\} $\rightarrow$ \{\textit{iY} $\cup$ \textit{RI}\}. Besides, the attribute values of the expanded items are not modified, and \textit{sup}(\{\textit{qX} $\cup$ \textit{LI}\}, $\mathcal{D}$) = \textit{sup}(\{\textit{iX} $\cup$ \textit{LI}\}, $\mathcal{D}^\prime$). Therefore, Lemma \ref{le3} holds.
\end{proof}

\begin{definition}[Unpromising item]
    \rm In a sequence database $\mathcal{D}$, if an item $e$ satisfies \textit{Attr}($e$, $\mathcal{D}$) $\textless$ \textit{minAttr}, it is called a global unpromising item. For a sequential rule $r$ = $X$ $\rightarrow$ $Y$, let \textit{lr} = \{$X$ $\cup$ $e$\} $\rightarrow$ $Y$. If the attribute values of \textit{lr} and all its expanded rules are less than \textit{minAttr}, then \textit{e} is the unpromising left-expanded item of $r$. Next, let \textit{rr} = $X$ $\rightarrow$ \{$Y$ $\cup$ $e$\}. If the attribute values of \textit{rr} and all its expanded rules are less than \textit{minAttr}, then \textit{e} is the unpromising right-expanded item of $r$.
\end{definition}

To obtain the target sequential rules more efficiently, the unpromising items need to be removed or identified as early as possible during the mining preparation or process. Operations involving these unpromising items do not affect the generation of correct results. This paper first introduces the upper bounds that can be used to estimate the maximum attribute value of items. Then, the related pruning strategies are proposed.

\begin{upper bound}[Basic upper bound] 
    \rm In the sequence $s$, the basic upper bound value of an item $i$ (denoted by \textit{UB}($i$, $s$)) must be greater than or equal to the maximum attribute value that this sequence can provide for rules. Thus, for an item $i$ in the database $\mathcal{D}$, its basic upper bound value can be denoted by \textit{UB}($i$, $\mathcal{D}$), and is defined as \textit{UB}($i$, $\mathcal{D}$) = $\sum\limits_{\{s \vert s \in \mathcal{D}\}}$ \textit{UB}($i$, $s$).
\end{upper bound}

\begin{theorem}
    \label{ub bub}
    \rm If an item $i$ satisfies \textit{UB}($i$, $\mathcal{D}$) $\textless$ \textit{minAttr}, then any rule $r$ containing $i$ has \textit{Attr}($r$, $\mathcal{D}$) $\textless$ \textit{minAttr}.
\end{theorem}
\begin{proof}
    \rm In a sequence $s$, for a sequential rule $r$ containing the item $i$, if the basic upper bound value of $r$ is less than \textit{minAttr}, then we can have that \textit{Attr}($r$, $s$) $\le$ \textit{UB}($i$, $s$) $\textless$ \textit{minAttr}. Thus, in a database $\mathcal{D}$, there is \textit{Attr}($r$, $\mathcal{D}$) $\le$ \textit{UB}($i$, $\mathcal{D}$) $\textless$ \textit{minAttr}. Therefore, Theorem \ref{ub bub} holds.
\end{proof}

The calculation of the basic upper bound can vary under rule mining tasks using different evaluation metrics. For example, \textit{SEU} is an upper bound that can be used in utility-driven tasks.

\begin{strategy}
    \label{ui bub}
    \rm According to Theorem \ref{ub bub}, for an item $i$, if \textit{UB}($i$, $\mathcal{D}^\prime$) $\textless$ \textit{minAttr}, then $i$ can be removed from $\mathcal{D}^\prime$.
\end{strategy}

Although the basic upper bound can be used to filter out unpromising items, tighter upper bounds can be proposed for targeted mining to remove more items from the database, thus improving efficiency.

\begin{upper bound}
    \rm Given a query rule \textit{qr} and a modified sequence database $\mathcal{D}^\prime$ (processed by Strategy \ref{isr} and Technique \ref{te rsm}), the left-expanded upper bound of an item $i$ in the database $\mathcal{D}^\prime$ can be denoted by \textit{TUB}$_L$($i$, $\mathcal{D}^\prime$), and is defined as follows.
    $$
    \begin{aligned}
        \textit{TUB}_L(i, \mathcal{D}^\prime) \;=\; \sum_{\{s \vert s \in \mathcal{D}^\prime \;\wedge\; \exists\; i.\textit{pos} \textless s_{py}\}} \textit{UB}(i, s),
    \end{aligned}
    $$
    where $i.\textit{pos}$ denotes the position of any itemset in the sequence that contains $i$.
\end{upper bound}

\begin{theorem}
\label{tubl}
    \rm Given a query rule \textit{qr} and a modified sequence database $\mathcal{D}^\prime$, if an item $i$ satisfies \textit{TUB}$_L$($i$, $\mathcal{D}^\prime$) $\textless$ \textit{minAttr}, then any expanded rule $r$ of \textit{qr} whose antecedent contains $i$ has \textit{Attr}($r$, $\mathcal{D}^\prime$) $\textless$ \textit{minAttr}.
\end{theorem}

\begin{proof}
    \rm For a sequence $s$ containing \textit{qr} = \textit{qX} $\rightarrow$ \textit{qy} in $\mathcal{D}^\prime$, it has $s_{px}$ and $s_{py}$. If an item $i$ is a left-expanded item of \textit{qr}, then it satisfies $i.\textit{pos}$ $\textless$ $s_{py}$. Thus, \textit{TUB}$_L$($i$, $\mathcal{D}^\prime$ is the maximum attribute value that sequences containing $t$ = \{\textit{qX} $\cup$ $i$\} $\rightarrow$ \textit{qy} can provide. If it is less than \textit{minAttr}, any expanded rule of $t$ has a low attribute value. Therefore, Theorem \ref{tubl} holds.
\end{proof}

\begin{strategy}
    \label{ui lub}
    \rm According to Theorem \ref{tubl}, for a query rule \textit{qr} and an item $i$, if \textit{TUB}$_L$($i$, $\mathcal{D}^\prime$) $\textless$ \textit{minAttr}, then $i$ is not used to perform the left-expansion.
\end{strategy}

\begin{upper bound}
    \rm Similarly, given a query rule \textit{qr} and a modified sequence database $\mathcal{D}^\prime$, the right-expanded upper bound of an item $i$ in the database $\mathcal{D}^\prime$ can be denoted by \textit{TUB}$_R$($i$, $\mathcal{D}^\prime$), and is defined as follows.
    $$
    \begin{aligned}
        \textit{TUB}_R(i, \mathcal{D}^\prime) \;=\; \sum_{\{s \vert s \in \mathcal{D} \;\wedge\; \exists\; i.\textit{pos} \textgreater s_{px}\}} \textit{UB}(i, s).
    \end{aligned}
    $$
\end{upper bound}

\begin{theorem}
\label{tubr}
    \rm Given a query rule \textit{qr} and a modified sequence database $\mathcal{D}^\prime$, if an item $i$ satisfies \textit{TUB}$_R$($i$, $\mathcal{D}^\prime$) $\textless$ \textit{minAttr}, then any expanded rule $r$ of \textit{qr} whose consequent contains $i$ has \textit{Attr}($r$, $\mathcal{D}^\prime$) $\textless$ \textit{minAttr}.
\end{theorem}

\begin{proof}
    \rm For a sequence $s$ containing \textit{qr} = \textit{qX} $\rightarrow$ \textit{qy} in $\mathcal{D}^\prime$, it has $s_{px}$ and $s_{py}$. If an item $i$ is a right-expanded item of \textit{qr}, then it satisfies $i.\textit{pos}$ $\textgreater$ $s_{px}$. Thus, \textit{TUB}$_R$($i$, $\mathcal{D}^\prime$) is the maximum attribute value that sequences containing $t$ = \textit{qX} $\rightarrow$ \{\textit{qy} $\cup$ $i$\} can provide. If it is less than \textit{minAttr}, any expanded rule of $t$ has a low attribute value. Therefore, Theorem \ref{tubr} holds.
\end{proof}

\begin{strategy}
    \label{ui rub}
    \rm According to Theorem \ref{tubr}, for a query rule \textit{qr} and an item $i$, if \textit{TUB}$_R$($i$, $\mathcal{D}^\prime$) $\textless$ \textit{minAttr}, then $i$ is not used to perform the right-expansion.
\end{strategy}

\begin{strategy}
    \label{ui rui}
    \rm According to Theorems \ref{tubl} and \ref{tubr}, for a query rule \textit{qr} and an item $i$, if \textit{TUB}$_L$($i$, $\mathcal{D}^\prime$) $\textless$ \textit{minAttr} and \textit{TUB}$_R$($i$, $\mathcal{D}^\prime$) $\textless$ \textit{minAttr} are both satisfied, then $i$ can be removed from $\mathcal{D}^\prime$.
\end{strategy}

\begin{technique}[Start mining from the query rule]
    \label{te smfqr}
    \rm Given a query rule \textit{qr} and a sequence database $\mathcal{D}$, after using Strategy \ref{isr} and Technique \ref{te rsm}, the modified sequence database $\mathcal{D}^\prime$ can be obtained. According to Lemmas \ref{le1}-\ref{le3}, we can know that any qualified target sequential rule of \textit{qr} contains \textit{iX} $\rightarrow$ \textit{iY}. Thus, the mining process of the algorithm can start from the construction of \textit{iX} $\rightarrow$ \textit{iY} to find the target rules, instead of generating 1 * 1 rules first \cite{fournier2015mining, zida2015efficient}.
\end{technique}

\begin{algorithm}[h]
    \small
	\caption{The main procedure}
	\label{AL:find tr}
	\LinesNumbered
	\KwIn{\textit{qr} = \textit{qX} $\rightarrow$ \textit{qY}: a query sequential rule; $\mathcal{D}$: a sequence database, \textit{minAttr}: a minimum attribute threshold, \textit{minConf}: a minimum confidence threshold.} 
	\KwOut{all target sequential rules of \textit{qr}.}

    remove irrelevant sequences from $\mathcal{D}$ (\textbf{Strategy} \ref{isr})\; 
    use \textbf{Technique} \ref{te rsm} to obtain $\mathcal{D}^\prime$\;
    remove unpromising items from $\mathcal{D}^\prime$ (\textbf{Strategy} \ref{ui bub})\; 
    
    build \textit{ftr} = \textit{iX} $\rightarrow$ \textit{iY}\;
    \If{\rm the size of \textit{qr} is larger than or equal to 1 * 1}{
        \If{\rm \textit{Attr}(\textit{ftr}) $\ge$ \textit{minAttr} \&\& \textit{conf}(\textit{ftr}) $\ge$ \textit{minConf}}{
            output \textit{ftr}\;
        }    
    }

    obtain \textit{LEI} and \textit{REI} (\textbf{Strategies} \ref{ui lub} and \ref{ui rub})\; 
    \If{\textit{qY} \rm is not empty}
    {
        \For{$i$ $\in$ $\textit{LEI}$ }{
            {\textbf{call}} \textbf{leftExpand}(\{\textit{iX} $\cup$ $i$\} $\rightarrow$ \textit{iY})\;
	    }
    }
    \For{$i$ $\in$ $\textit{REI}$ }{
        {\textbf{call}} \textbf{rightExpand}(\textit{iX} $\rightarrow$ \{\textit{iY} $\cup$ $i$\})\;
	}	
\end{algorithm}

\begin{algorithm}[h]
    \small
	\caption{Left-expanding sequential rules}
	\label{AL:left tr}
	\LinesNumbered
	\KwIn{\textit{pr} = \textit{pX} $\rightarrow$ \textit{pY}: a potential target sequential rule.} 
	\KwOut{all qualified target sequential rules.}

    \If{\textit{qX} \rm is not empty $\vert\vert$ the size of \textit{pX} is larger than 1}{
        \If{\rm \textit{Attr}(\textit{pr}) $\ge$ \textit{minAttr} \&\& \textit{conf}(\textit{pr}) $\ge$ \textit{minConf}}{
            output \textit{pr}\;
        }  
    }
    obtain \textit{LEI} (\textbf{Strategy} \ref{ui lub})\;
        \For{$i$ $\in$ $\textit{LEI}$ }{
            {\textbf{call}} \textbf{leftExpand}(\{\textit{pX} $\cup$ $i$\} $\rightarrow$ \textit{pY})\;
	    }
\end{algorithm}

\begin{algorithm}[h]
    \small
	\caption{Right-expanding sequential rules}
	\label{AL:right tr}
	\LinesNumbered
	\KwIn{\textit{pr} = \textit{pX} $\rightarrow$ \textit{pY}: a potential target sequential rule.} 
	\KwOut{all qualified target sequential rules.}

    \If{\textit{qX} \rm is not empty $\vert\vert$ the size of \textit{pX} is larger than 1}{
        \If{\textit{qY} \rm is not empty $\vert\vert$ the size of \textit{pY} is larger than 1}{
            \If{\rm \textit{Attr}(\textit{pr}) $\ge$ \textit{minAttr} \&\& \textit{conf}(\textit{pr}) $\ge$ \textit{minConf}}{
                output \textit{pr}\;
            }       
        }
    }
    obtain \textit{LEI} and \textit{REI} (\textbf{Strategy} \ref{ui lub} and \textbf{Strategy} \ref{ui rub})\; 
    \If{\textit{qY} \rm is not empty $\vert\vert$ the size of \textit{pY} is larger than 1}
    {
        \For{$i$ $\in$ $\textit{LEI}$ }{
            {\textbf{call}} \textbf{leftExpand}(\{\textit{pX} $\cup$ $i$\} $\rightarrow$ \textit{pY})\;
	    }
    }
    \For{$i$ $\in$ $\textit{REI}$ }{
        {\textbf{call}} \textbf{rightExpand}(\textit{pX} $\rightarrow$ \{\textit{pY} $\cup$ $i$\})\;
	}
\end{algorithm}

The pseudocode of the main procedure is shown in Algorithm \ref{AL:find tr}. It demonstrates how to quickly discover target sequential rules. The original database $\mathcal{D}$ is processed by Strategy \ref{isr} and Technique \ref{te rsm} to obtain $\mathcal{D}^\prime$ (lines 1-2). Then, those unpromising global items are removed from  $\mathcal{D}^\prime$ and the first target sequential rule \textit{frt} is built (lines 3-4). If \textit{ftr} is a qualified target sequential rule of \textit{qr} (satisfying the required constraints), it is outputted (lines 5-8). Next, \textit{LEI} and \textit{REI} are obtained, which record the left- and right-expanded items, respectively (line 10). Note that strategies \ref{ui lub} and \ref{ui rub} can be used to remove useless expanded items. Finally, the left- / right-expansion procedure (shown in Algorithms \ref{AL:left tr} and \ref{AL:right tr}) is executed (lines 11-18). If the consequent of the query rule is empty, then \textit{ftr} needs to be right-extended first. This is because, once a rule has been left-expanded, the right-expansion cannot be performed in the algorithm. If this requirement is not used, many invalid rules of size no larger than or equal to 1 * 1 will be generated. In the left expanded procedure of the sequential rule \textit{pr}, if it is a qualified target sequential rule, then it is outputted (lines 1-5). If the antecedent of the query rule is empty, then \textit{pr} needs to be left-expanded at least once. Then, left-expanded items are obtained, and the associated expansions are also performed recursively (lines 6-9). The main operations of the right-expansion of the sequential rule are similar to its left-expansion. The algorithm first determines whether \textit{pr} is a qualified target sequential rule and then performs the relevant expansion on it (lines 1-16). Note that the left-expansion can only be performed if the consequent of the rule corresponding to \textit{pr} is not empty.

\section{The Most Useful Sequential Rule Search}

In a mining task, the generation of target sequential rules depends mainly on the query rule set by the user. After numerous results are obtained, an important question is which of these rules is more valuable. For example, if the query rule is set as \{$a$\} $\rightarrow$ \{$b$\}, and there is a rule \{$a$, $c$, $d$\} $\rightarrow$ \{$b$\} with a high attribute value and a rule \{$a$\} $\rightarrow$ \{$b$, $d$\} with a low attribute value, which rule is more valuable? To measure the importance of the target rule with respect to the query rule, based on several commonly used similarity coefficients, we propose two metrics: the target rule Jaccard similarity metric (TRJS) and the target rule dice similarity metric (TROS). Note that TRJS is modified from the Jaccard similarity coefficient, while TROS is modified from the Dice similarity coefficient.

\begin{definition}[Target rule jaccard similarity metric, TRJS]
    \rm Given a query rule \textit{qr} and a target rule \textit{tr} of it, the TRJS value of \textit{tr} is denoted by \textit{TRJS}(\textit{qr}, \textit{tr}), and is defined as:
    $$
        \textit{TRJS}(\textit{qr}, \textit{tr}) = \frac{\vert \mathcal{D}_{\textit{qr}} \cap \mathcal{D}_{\textit{tr}}\vert}{\vert \mathcal{D}_{\textit{qr}} \cup \mathcal{D}_{\textit{tr}}\vert}
        = \frac{\textit{sup}(\textit{tr}, \mathcal{D})}{\textit{sup}(\textit{qr}, \mathcal{D})}.
    $$
    Since the support value of \textit{tr} is greater than 0 and not greater than the support value of \textit{qr}, the TRJS value of \textit{tr} is between (0, 1]. 
\end{definition}

\begin{theorem}
    \rm Given two target rules \textit{tr}$_{a}$ and \textit{tr}$_{b}$, if \textit{tr}$_{b}$ is larger than \textit{tr}$_{a}$, then \textit{TRJS}(\textit{qr}, \textit{tr}$_{b}$) $\le$ \textit{TRJS}(\textit{qr}, \textit{tr}$_{a}$).
\end{theorem}

\begin{proof}
    \rm If \textit{tr}$_{b}$ is larger than \textit{tr}$_{a}$, then we can have that \textit{sup}(\textit{tr}$_{b}$, $\mathcal{D}$) $\le$ \textit{sup}(\textit{tr}$_{a}$, $\mathcal{D}$) and thus \textit{TRJS}(\textit{qr}, \textit{tr}$_{b}$) $\le$ \textit{TRJS}(\textit{qr}, \textit{tr}$_{a}$).
\end{proof}

Therefore, if the TRJS value of a sequential rule is smaller than the minimum TRJS threshold (\textit{minTRJS}) set by users, then all of its expansion operations can be terminated.

\begin{definition}[Target rule dice similarity metric, TROS]
    \rm Given a query rule \textit{qr} and a target rule \textit{tr} of it, the TROS value of \textit{tr} is denoted by \textit{TROS}(\textit{qr}, \textit{tr}), and is defined as:
    $$
        \textit{TROS}(\textit{qr}, \textit{tr}) = \frac{2 \;*\; \vert \mathcal{D}_{\textit{qr}} \cap \mathcal{D}_{\textit{tr}}\vert}{
        \vert \mathcal{D}_{\textit{qr}} \vert \;+\; \vert \mathcal{D}_{\textit{tr}}\vert}
        = \frac{2 \;*\;\textit{sup}(\textit{tr}, \mathcal{D})}{\textit{sup}(\textit{qr}, \mathcal{D}) \;+\; \textit{sup}(\textit{tr}, \mathcal{D})}.
    $$
    Since the support value of \textit{tr} is greater than 0 and not greater than the support value of \textit{qr}, the TRJS value of \textit{tr} is between (0, 1]. 
\end{definition}

\begin{theorem}
    \label{th:tros}
    \rm Given two target rules \textit{tr}$_{a}$ and \textit{tr}$_{b}$, if \textit{tr}$_{b}$ is large than \textit{tr}$_{a}$, then \textit{TROS}(\textit{qr}, \textit{tr}$_{b}$) $\le$ \textit{TROS}(\textit{qr}, \textit{tr}$_{a}$).
\end{theorem}

\begin{proof}
    \rm If \textit{tr}$_{b}$ is larger than \textit{tr}$_{a}$, then we can have that \textit{sup}(\textit{tr}$_{b}$, $\mathcal{D}$) $\le$ \textit{sup}(\textit{tr}$_{a}$, $\mathcal{D}$) and \textit{sup}(\textit{tr}$_{b}$, $\mathcal{D}$) * (\textit{sup}(\textit{qr}, $\mathcal{D}$) + \textit{sup}(\textit{tr}$_{a}$, $\mathcal{D}$)) $\le$ \textit{sup}(\textit{tr}$_{a}$, $\mathcal{D}$) * (\textit{sup}(\textit{qr}, $\mathcal{D}$) + \textit{sup}(\textit{tr}$_{b}$, $\mathcal{D}$)). Theorem \ref{th:tros} holds.
\end{proof}

\begin{algorithm}[!ht]\small
	\caption{Judging useful sequential rules}
	\label{AL:ju tr}
	\LinesNumbered
	\KwIn{\textit{pr} = \textit{pX} $\rightarrow$ \textit{pY}: a potential target sequential rule; \textit{qrSup}: the support value of the query rule; \textit{minRS}: the minimum rule similarity metric threshold.} 
	\KwOut{Whether \textit{pr} is a useful rule or not.}
 
    set \textit{iSup} as the support value of \textit{pr}\;
    calculate TRJS / TROS of \textit{pr} based on \textit{iSup} and \textit{qrSup}\;
    \If{\rm{TRJS / TROS of} \textit{pr} $\textless$ \textit{minRS}}{
        return \textit{false};
    }
    \textbf{return} \textit{true}
\end{algorithm}

Algorithm \ref{AL:ju tr} judges whether a rule is a useful rule or not. It has three input parameters: \textit{pr}, \textit{qrSup}, and \textit{minRS}. The \textit{qrSup} is computed before mining the target sequential rules, i.e., after filtering irrelevant sequences and items. The judgment procedure first obtains the variable \textit{iSup}, which is the support value of \textit{pr} (line 1). Then, based on \textit{iSup} and \textit{qrSup}, the TRJS / TROS value of \textit{pr} can be easily computed (line 2). Finally, the judgment procedure is completed (lines 3-6).

\section{Experiments} \label{sec:experiments}

In this section, to fully evaluate the proposed techniques and methods, we performed extensive experiments on both real-world and synthetic datasets and provided an in-depth analysis and discussion.

\subsection{Experimental Setup and Datasets}

We applied the techniques and pruning strategies proposed in this paper to benchmark algorithms for different mining tasks and implemented the following algorithms for our experiments:

\begin{itemize}
    \item \textbf{FSSR}: An implementation of frequent target sequential rule search modified from RuleGrowth \cite{fournier2015mining}, incorporating the proposed techniques and pruning strategies.
    
    \item \textbf{FSUR}: This approach utilizes the proposed techniques and pruning strategies for targeted utility mining tasks. It is modified based on the US-Rule algorithm \cite{huang2023us}.
    
    \item \textbf{{FSUR}$_\textit{filter}$}: A simple variant only employs Technique \ref{te rsm}.
    
    \item \textbf{{FSUR}$_\textit{SEU}$}: Compared to FSUR, this approach utilizes the upper bound SEU to filter unpromising items.
    
    \item \textbf{{TRJS}$_i$} and \textbf{{TRJO}$_i$}: Two methods using the corresponding proposed similarity metrics TRJS and TRJO, respectively, where $i$ represents the minimum rule similarity metric threshold.
\end{itemize}

In the experiments, we used six datasets, four of which are real-world datasets. These datasets are generated primarily from text-related materials, such as books, novels, and sign language. Table \ref{dataset} presents the major characteristics of these datasets, including the dataset size, the number of items, the average length of sequences (AvgS), and the average size of itemsets (AvgIT). Detailed descriptions of the datasets can be found in Refs. \cite{gan2019survey,gan2020fast}. In addition, we randomly selected query rules that yield non-empty mining results for each experiment. To avoid particularly short runtimes for each algorithm, query rules with smaller sizes were selected.

All algorithms were implemented in Java. The experiments were conducted on a computer with a 64-bit Windows 11 operating system, equipped with a 12$^{\text{th}}$ Gen Intel(R) Core™ i7-12700F CPU and 16 GB RAM. The code and datasets are publicly available at \href{https://github.com/DSI-Lab1/FSUR}{https://github.com/DSI-Lab1/FSUR.}  

\begin{table}[h]
    \centering
    \caption{Dataset characteristics}
    \label{dataset}
    \renewcommand\arraystretch{1.3}
    \begin{tabular}{lccccc}
        \hline
        {\textbf{Dataset}} & {$\vert\mathcal{D}\vert$} & {$\vert\textit{I}\vert$} & {\textbf{AvgS}} & {\textbf{AvgIT}} & {\textbf{Type}} \\ \hline
           Bible & 36,369& 13,905& 21.64& 1.0& Book       \\
           Sign & 730 &  267& 51.99& 1.0 & Sign-language  \\
           Kosarak10k & 10,000 & 10,094& 8.14 & 1.0 & News \\
           Leviathan & 5,834& 9,025& 33.81 & 1.0& Novel\\
           Syn20K & 20,000& 7,442& 26.97& 4.33 &  Synthetic data\\
           Syn40K & 40,000& 7,537& 26.84 & 4.33 & Synthetic data\\ \hline
    \end{tabular}
\end{table}

\begin{table}[!ht] 
\fontsize{4.3pt}{10pt}\selectfont
    \centering
    \caption{The number of expansions under various \textit{minutil}}
    \renewcommand{\arraystretch}{1.5}
    \label{expansions}
    \begin{tabular}{|c|c|c|c|c|c|c|c|}
    \hline
    \multicolumn{1}{|c|}{\textbf{Dataset}} & \textbf{Algorithm} & \textit{minutil}$_{1}$ & \textit{minutil}$_{2}$ & $\textit{minutil}_{3}$ & $\textit{minutil}_{4}$ & $\textit{minutil}_{5}$ & $\textit{minutil}_{6}$\\
         \hline\hline
         \multirow{3}{*}{\textbf{Bible}}& FSUR$_\textit{filter}$ & 62,798,797
 & 22,995,858 & 9,006,470 & 4,318,620 & 2,779,546 & 2,240,172\\ 
         \cline{2-8} & FSUR$_\textit{SEU}$ & 3,391,806 & 1,269,186 & 487,726 & 212,718 & 121,433 & 91,655\\  
        \cline{2-8} & FSUR & 2,637,005 & 930,607 & 346,783 & 156,285 & 96,564 & 77,623\\ \hline
        
         \multirow{3}{*}{\textbf{Sign}}& FSUR$_\textit{filter}$ & 4,111,100
 & 3,021,123 & 2,268,195 & 1,767,893 & 1,365,701 & 1,080,415\\ 
         \cline{2-8} & FSUR$_\textit{SEU}$ & 1,276,604 & 937,297 & 703,953 & 549,589 & 426,441 & 337,893\\  
        \cline{2-8} & FSUR & 1,230,963 & 888,150 & 683,045 & 533,029 & 422,197 & 332,894\\ \hline 
        
         \multirow{3}{*}{\textbf{Kosarak10k}}& FSUR$_\textit{filter}$ & /
 & 67,026 & 38,603 & 32,155 & 12,714 & 3,403\\ 
         \cline{2-8} & FSUR$_\textit{SEU}$ & / & 2,922 & 2,500 & 2,070 & 879 & 262\\  
        \cline{2-8} & FSUR & 3,151 & 2,858 & 2,377 & 889 & 251 & 23\\ \hline 
        
         \multirow{3}{*}{\textbf{Leviathan}}& FSUR$_\textit{filter}$ & 14,031,353
 & 7,349,344 & 4,438,927 & 2,987,633 & 2,095,266 & 1,557,083\\ 
         \cline{2-8} & FSUR$_\textit{SEU}$ & 3,827,557 & 2,054,881 & 1,266,066 & 861,905 & 611,609 & 456,904\\  
        \cline{2-8} & FSUR & 3,695,937 & 1,991,585 & 1,213,004 & 821,673 & 590,463 & 438,149\\ \hline 
        
         \multirow{3}{*}{\textbf{Syn20K}}& FSUR$_\textit{filter}$ & /
 & 118,164,086 & 97,369,042 & 62,310,552 & 28,168,055 & 9,440,853 \\ 
         \cline{2-8} & FSUR$_\textit{SEU}$ & 61,109,801 & 26,462,866 & 23,202,752 & 16,049,988 & 7,585,373 & 2,387,616\\  
        \cline{2-8} & FSUR & 60,336,766 & 26,440,154 & 23,131,286 & 15,938,227 & 7,585,373 & 2,387,616 \\ \hline 
        
         \multirow{3}{*}{\textbf{Syn40K}}& FSUR$_\textit{filter}$ & 11,346,960
 & 3,053,393 & 270,423 & 97,519 & 84,266 & 25,441\\ 
         \cline{2-8} & FSUR$_\textit{SEU}$ & 3,140,237 & 751,479 & 219,124 & 71,617 & 22,197 & 6,674\\  
        \cline{2-8} & FSUR & 3,140,237 & 751,479 & 219,124 & 71,617 & 22,197 & 6,674 \\ \hline 
         \hline
    \end{tabular}
\end{table}

\begin{figure*}[!h]
	\centering
    \includegraphics[trim=0 0 0 0,clip,width=18.2cm, height=7.3cm]{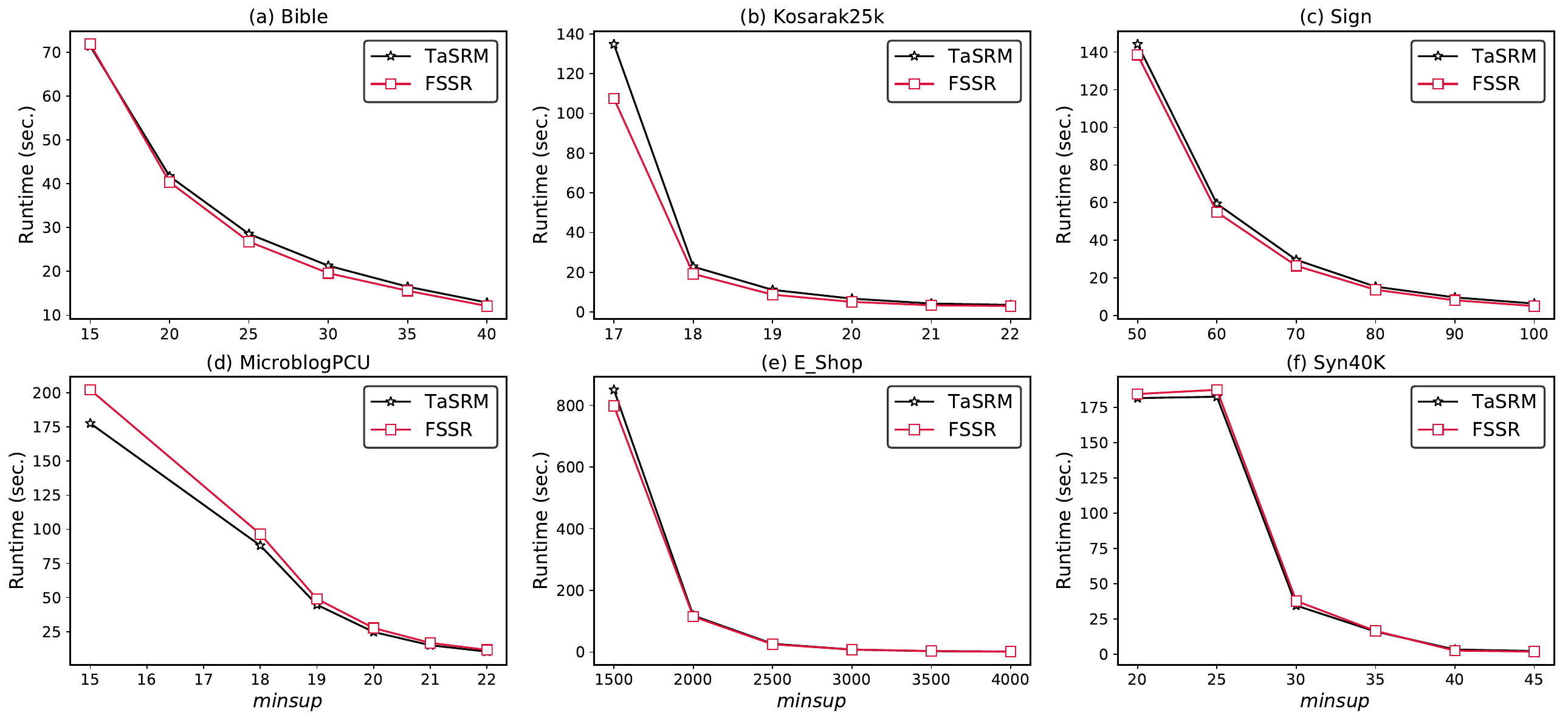}
    \caption{Runtime under different FPM tasks.}
    \label{mRuntime}
\end{figure*}

\begin{figure*}[!h]
	\centering
	\includegraphics[trim=0 0 0 0,clip,width=18.2cm, height=7.3cm]{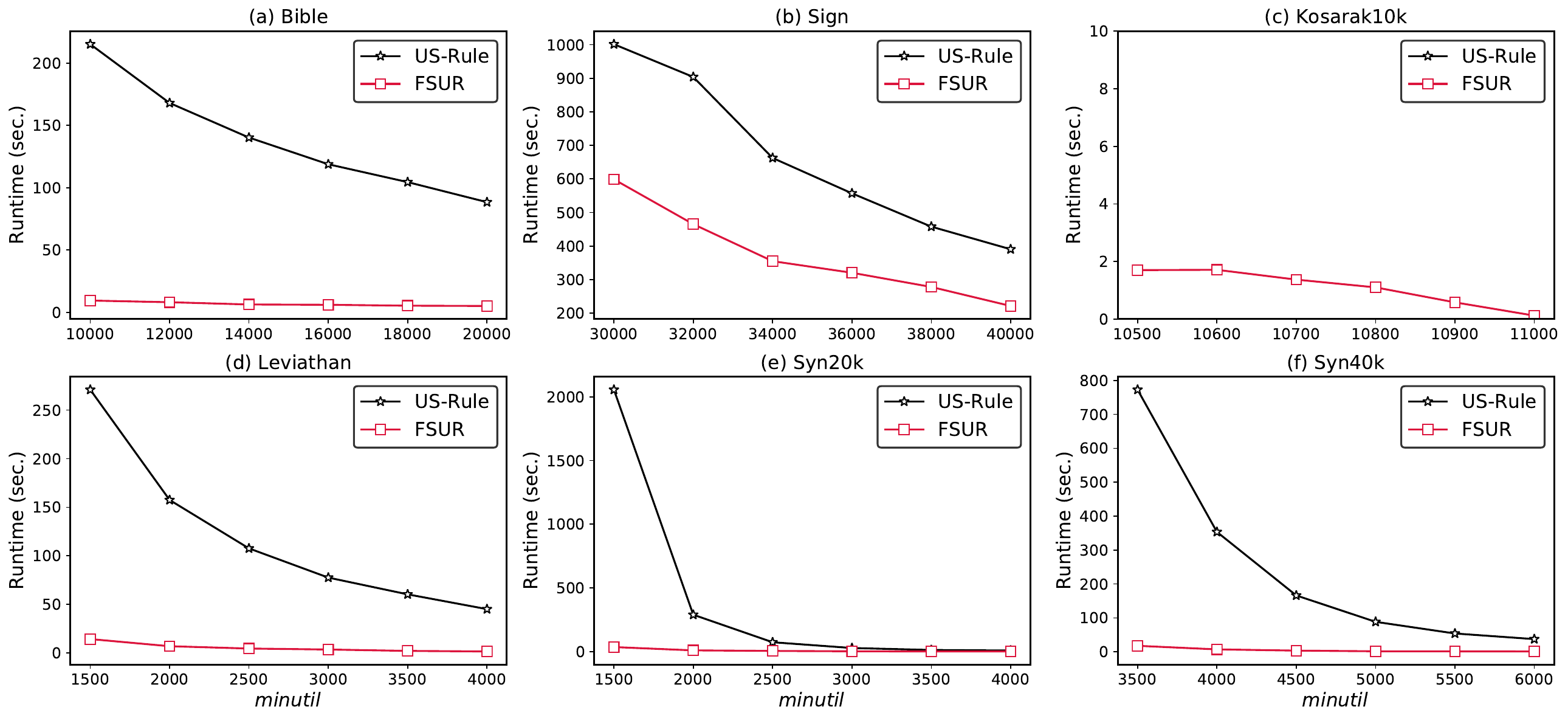}
    \caption{Runtime under different UPM tasks.}
    \label{uRuntime}
\end{figure*}

\subsection{Efficiency Analysis}

We first conducted experiments to analyze the efficiency of the proposed comprehensive search solution (which integrates all designed techniques and pruning strategies) across different sequence mining tasks. The runtime performance for FPM and UPM tasks is detailed in Fig. \ref{mRuntime} and Fig. \ref{uRuntime}, respectively. The results for FPM tasks, as shown in Fig. \ref{mRuntime}(a-f), demonstrate that our proposed algorithm FSSR can achieve efficiency comparable to TaSRM \cite{gan2025towards} (the state-of-the-art targeted sequential rule mining algorithm). For instance, on the Bible dataset (Fig. \ref{mRuntime}(a)), as the minimum support threshold (\textit{minsup}) increases from 15 to 40, the runtime of both algorithms decreases steadily from approximately 70 seconds to about 10 seconds. The performance curves of FSSR and TaSRM are nearly identical across all tested datasets (Bible, Kosarak25k, Sign, MicroblogPCU, E-shop, and Syn40K), indicating that FSSR can successfully match the efficiency of a highly optimized baseline. This is noteworthy because TaSRM is specifically designed for FPM. The comparable performance of FSSR confirms that our generalized search framework, when applied to the FPM task, incurs minimal overhead and retains high effectiveness.

For the more complex UPM tasks, the advantage of our proposed solution becomes significantly more pronounced. As shown in Fig. \ref{uRuntime}(a-f), our algorithm FSUR consistently and substantially outperforms the recent baseline US-Rule across all datasets and parameter settings. On smaller datasets like Sign (Fig. \ref{uRuntime}(b)), FSUR completes the mining task in roughly half the time required by US-Rule. This efficiency gap widens dramatically on larger or denser datasets. Besides, on the Kosarak10k dataset (Fig. \ref{uRuntime}(c)), US-Rule fails to finish tasks within a reasonable runtime for most \textit{minutil} values, as indicated by its missing or exceedingly high data points on the figures. On the contrary, FSUR maintains a stable runtime throughout the parameter range. This demonstrates FSUR's robustness and practical applicability where the baseline fails.

In summary, the experimental results clearly show that our proposed search solution is both effective and highly efficient. It achieves parity with the state-of-the-art in FPM tasks and possesses superior—often drastically better—performance in UPM tasks. This efficiency stems fundamentally from the design of our pruning strategies, which successfully constrain the explosive search space inherent in sequence mining, resulting in faster execution and enhanced scalability.

\begin{figure*}[!ht]
    \centering
    \includegraphics[trim=0 0 0 0,clip,width=18.2cm, height=7.3cm]{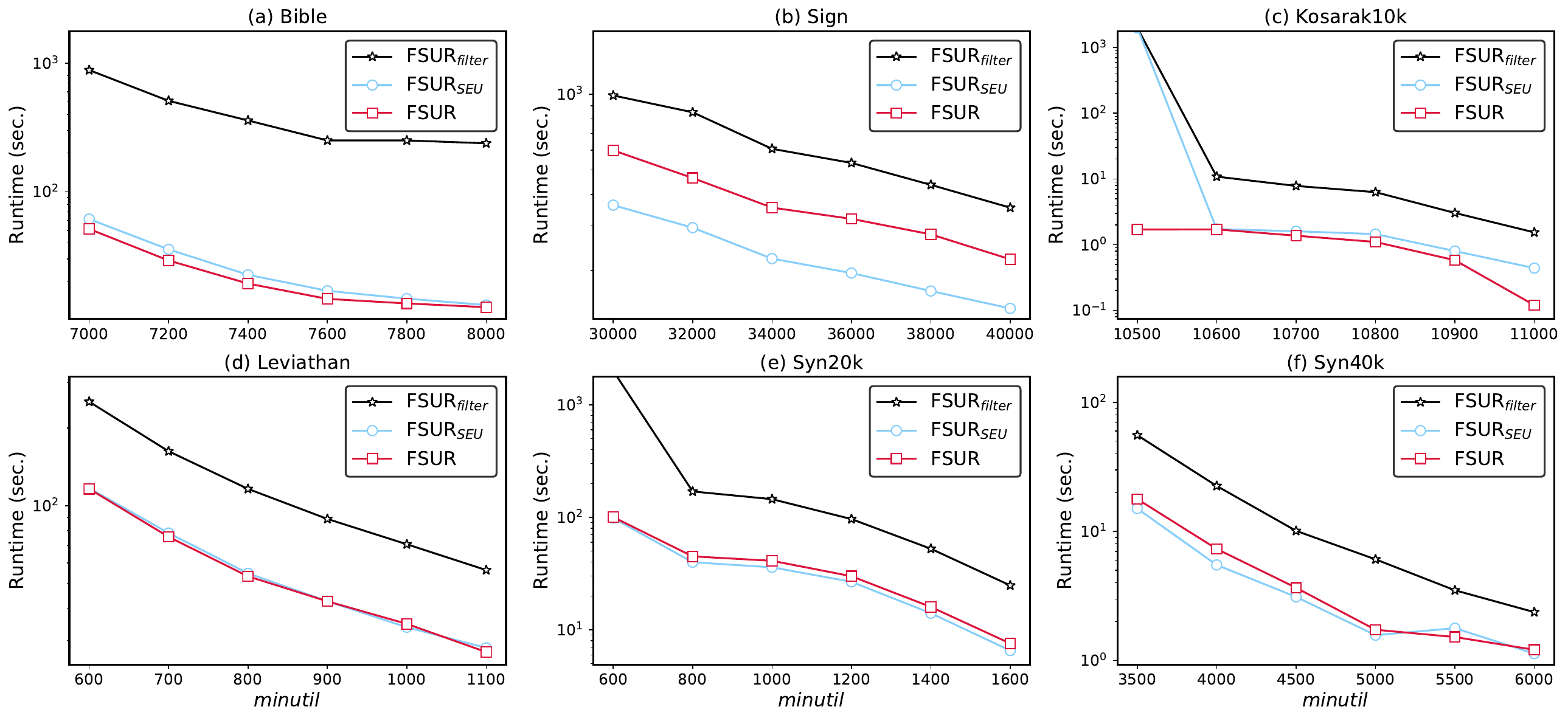}
    \caption{Runtime for different algorithm variants.}
     \label{sRuntime}
\end{figure*}

\subsection{Technique Analysis}

In this subsection, considering that the utility-driven mining process is more challenging, we further analyzed the main techniques proposed under utility mining tasks. We compared the performances of FSUR and its two variants (including {FSUR}$_\textit{filter}$ and {FSUR}$_\textit{SEU}$). The runtime for different algorithm variants is reported in Fig. \ref{sRuntime} and the number of expansions is shown in Table \ref{expansions} (\textit{minutil}$_{i}$ denotes the $i$-th smallest minimum utility threshold setting under this task). A comparison of {FSUR}$_\textit{filter}$ and FSUR can demonstrate the effect of Technique \ref{te smfqr}, since the use of Technique \ref{te rsm} allows both algorithms to process the same data size. As for the comparison between {FSUR}$_\textit{SEU}$ and FSUR, it is to analyze the specific performance of the basic upper bound and the tighter upper bounds.

From the experimental results, we know that FSUR always outperforms {FSUR}$_\textit{filter}$ by about an order of magnitude in terms of runtime. In general, on the Bible dataset, the number of expansions of FSUR is at least 20 times less than the number of expansions of {FSUR}$_\textit{filter}$. On the other datasets, it is typically about 4 times less. This demonstrates that Technique \ref{te smfqr} can reduce runtime by avoiding a large number of pointless expansions at a low cost. In addition, the experiments also reveal some situations regarding the basic upper bound and the tighter upper bounds. Except on the datasets Sign and Kosarak10k (\textit{minutil} set to 10,500), the difference in runtime consumed by FSUR and {FSUR}$_\textit{SEU}$ is not significant. The use of tighter upper bounds leads to a reduction in the number of expansions during the mining process, but this improvement is not considerable on the chosen datasets. In particular, on the synthetic datasets, we can see that the tighter upper bounds generally do not reduce the number of expansions, and they require extra time for their computation. Since the Bible is a dataset with more items and is less dense, the number of expansions of {FSUR} is much less compared to {FSUR}$_\textit{SEU}$. Note that the tighter upper bounds achieve the largest negative effect on the Sign dataset compared to the other datasets. This is because the dense Sign dataset contains an exceptionally small number of distinct items, and the sequences in Sign are all very long. Therefore, the utilization of tighter upper bounds is even less cost-effective.

\subsection{Effect of Target Rule Similarity Metric}

Finally, we analyze the effect of the two similarity metrics proposed in this paper on mining results and efficiency. Three minimum similarity metric threshold values for the metrics TRJS and TROS are selected: 0.001, 0.01, and 0.05. Runtimes and the number of target rules obtained for different algorithm variants on metrics TRJS and TROS are reported in Fig. \ref{rsRuntime} and Fig. \ref{oRuntime}, respectively. Experiments show that the proposed metrics can help users to flexibly adjust mining tasks. From the results, we know that the higher the minimum similarity metric threshold is set, the less time the algorithm consumes, and the fewer target rules are obtained. The number of obtained rules does not increase much with decreasing \textit{minutil}, even if the TRJS metric is set to a very low value. The use of TRJS stabilizes the runtime of each algorithm. The target rules obtained by {TRJS}$_{0.01}$ and {TRJS}$_{0.001}$ are within 100 and 200 rules, respectively. However, when \textit{minutil} is reduced from 7,000 to 6,800, the number of target rules obtained by FUSR increases substantially. This situation makes it difficult to ensure which rules are worth prioritizing for users to analyze. For the metric TROS, when the parameter \textit{minutil} is set to 7000 or more, all algorithms obtain around 200 target rules. In the experiments, FUSR and {TROS}$_{0.001}$ generate the same number of target rules, but {TROS}$_{0.001}$  filters out some irrelevant expansions with the help of the similarity metric. Compared to TRJS, the use of TROS does not cause the algorithm to discover only a few rules when the minimum similarity metric threshold is set to 0.05 or smaller values.

\begin{figure}[h]
    \centering
    \includegraphics[trim=0 0 0 0,clip,scale=0.3]{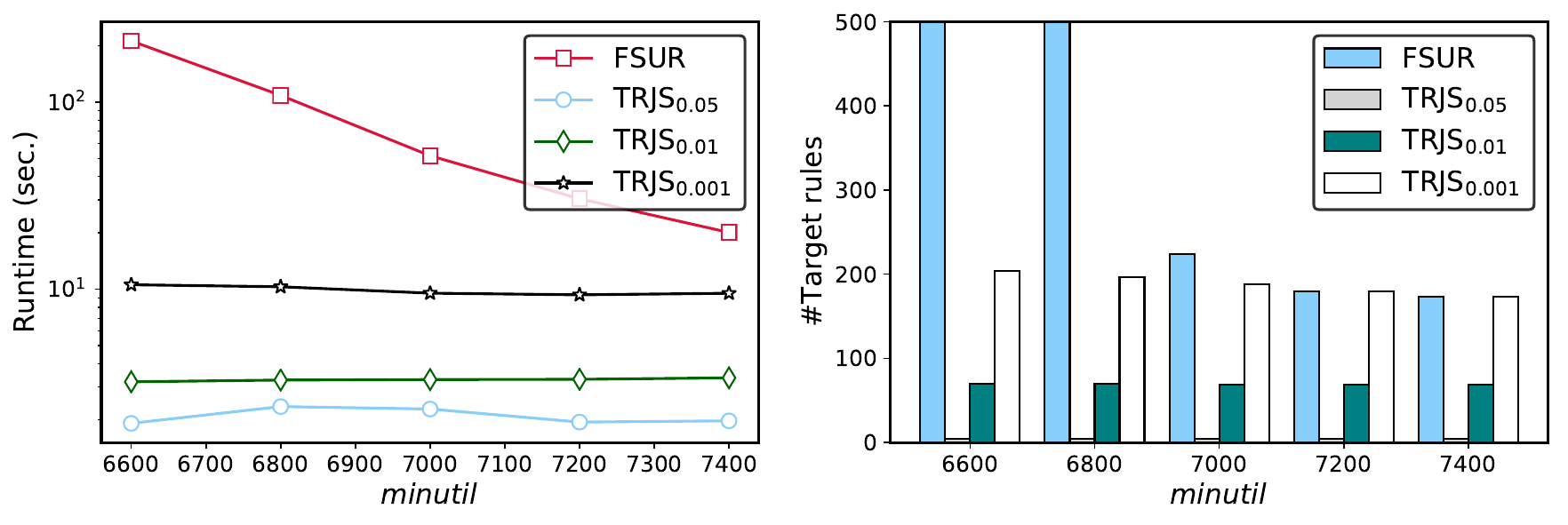}
    \caption{Runtime for different algorithm variants.}
    \label{rsRuntime}
\end{figure}

\begin{figure}[h]
    \centering
    \includegraphics[trim=0 0 0 0,clip,scale=0.3]{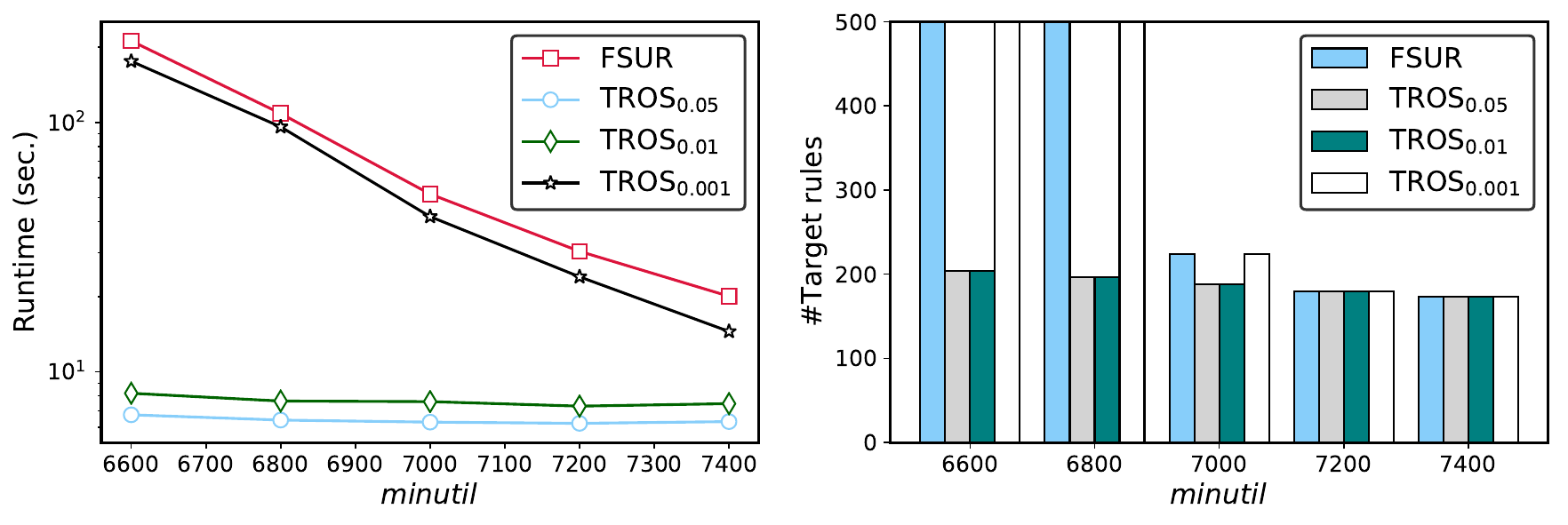}
    \caption{Runtime for different algorithm variants.}
    \label{oRuntime}
\end{figure}

\section{Conclusion}  \label{sec:conclusion}

In this paper, we study the problem of targeted mining of sequential rules under different mining metrics. We observe the characteristics of a query rule in a dataset and propose feasible techniques to process the dataset so that the size of the mined dataset can be reduced. Considering that a database has many unpromising items that cannot be used to construct target sequential rules, we introduce the basic upper bound and the derived tighter upper bounds to remove them. Furthermore, we propose two similarity metrics to discover the most relevant target rules based on the query rule. The experiments fully evaluate the techniques, pruning strategies, and the target rule similarity metrics proposed in this paper. There remains much interesting work to explore. It is still meaningful to consider how to develop complex targeted query settings and new evaluation metrics for sequential rules to further improve the quality of mining results. As there are currently few studies on targeted sequential rule discovery, some of the challenging practical scenarios include online rule search, privacy-preserving rule mining, and obtaining target sequential rules from uncertain data. Addressing these open problems will require continuous effort and innovation from the research community. Finally, we hope that the insights provided in this work can serve as a solid foundation and inspire more studies to delve into this promising field.

% references section
\bibliographystyle{IEEEtran}
\bibliography{FSUR}

\end{document}